\documentclass[12pt]{extarticle}
\usepackage{extsizes}
\usepackage[reqno]{amsmath} 
\usepackage{amsmath}
\usepackage{amssymb}
\usepackage{amsthm}
\usepackage{amscd}
\usepackage{amsfonts}
\usepackage{graphicx}
\usepackage{dsfont}
\usepackage{fancyhdr}
\usepackage{enumerate}
\usepackage{youngtab}
\usepackage{braket}
\usepackage[ruled,vlined]{algorithm2e}
\usepackage[utf8]{inputenc}
\usepackage{comment}
\usepackage{cite}
\usepackage{tikz}
\usepackage{subcaption}

\usetikzlibrary{arrows.meta}

\newcommand{\Description}[1]{}

\usepackage[margin=1.5in]{geometry}
\usepackage{graphicx}
\usepackage{color}
\usepackage[hidelinks]{hyperref}

\theoremstyle{theorem}
\newtheorem{corollary}{Corollary}[section]

\newtheorem{proposition}[corollary]{Proposition}
\newtheorem{theorem}[corollary]{Theorem}

\begin{document}

\title{A SWAP-free Framework for QAOA}
\author{Thiago Assis\thanks{TA, PB, LL, GC at the Department of Computer Science, UFMG - Belo Horizonte, Brazil. [gabriel, thiagoassis, lailalopes, pedro.baptista@dcc.ufmg.br]. DF at Inter Science - Belo Horizonte, MG, Brazil. [diego.ferreira@inter.co]} \and Pedro Baptista\footnotemark[1] \and Laila Lopes\footnotemark[1] \and Diego Ferreira\thanks{DF at Inter Science - Belo Horizonte, MG, Brazil. [diego.ferreira@inter.co]} \and Gabriel Coutinho\footnotemark[1]}
\date{\today}
\maketitle
\vspace{-0.8cm}

\begin{abstract}
  The performance of the Quantum Approximate Optimization Algorithm (QAOA) on noisy intermediate-scale quantum (NISQ) devices is strongly limited by sparse qubit connectivity. When interactions required by QAOA Hamiltonians are not aligned to the hardware topology, transpilation introduces SWAP gates, increasing circuit depth and noise. We propose a SWAP-free QAOA framework based on modifying the cost Hamiltonian so that it can be implemented natively on the hardware. We formulate this as a mixed-integer semidefinite program (MISDP) that selects a hardware-compatible approximation of the original cost matrix and optimizes the allocation of logical variables to physical qubits. We prove that the associated decision problem is NP-complete and derive theoretical guarantees relating the MISDP objective to the loss in the original optimization problem through the Lov\'asz number of the hardware graph. Since solving MISDPs is practical only for small instances, we introduce heuristics based on spectral properties of the problem matrix and hardware graph. Our experiments on a cardinality-constrained quadratic optimization model for index tracking show competitive performance against a baseline representing ideal QAOA under SWAP-induced noise. These results indicate that, on sparse NISQ architectures, a hardware-aware approximation of the objective may be more effective than an exact but heavily transpiled Hamiltonian implementation.
\end{abstract}

\begin{center}
\end{center}

\section{Introduction}

In the near term, one of the most anticipated applications of quantum computing is solving Quadratic Unconstrained Binary Optimization (QUBO) problems, an NP-hard class with applications in many areas, including finance, logistics and chemistry \cite{qubo-ising,boulebnane2023peptide, qubo-applications, quantum-finance-survey}. A prominent quantum approach for the solution of such problems is the Quantum Approximate Optimization Algorithm (QAOA), which seeks high-quality approximate solutions through a variational quantum-classical loop \cite{qaoa,qaoa-analysis}.

However, one of the principal barriers to using QAOA lies in the restrictive nature of today's Noise Intermediate Scale Quantum (NISQ) hardware \cite{nisq}. When qubit interactions required by the QAOA circuit do not align with the hardware's topology, the quantum transpiler must insert SWAP operations to route qubits, increasing the circuit depth and consequently the accumulated error. This issue is exacerbated by the sparse connectivity of current quantum hardware, which will usually require a quadratic number of SWAP gates for QAOA transpilation even under optimized settings \cite{hashim2022optimized}.

To overcome this hardware-imposed limitation, we introduce a framework for the design of SWAP-free QAOA circuits. The core idea is to modify the QAOA Hamiltonian such that it only includes interactions between qubits that are directly connected in the hardware connectivity graph, at the cost of potentially reduced solution quality. This approach eliminates the need for SWAP gates and thus reduces circuit depth and error accumulation, and despite the fact that the modified Hamiltonian is not an exact representation of the original problem, our experimental results show that this can lead to improved results on NISQ devices.

Our approach is formalized through a Mixed-Integer Semidefinite Programming (MISDP) formulation (see refs. \cite{de2024integrality, gally2018framework}), which optimizes the choice of interactions in the QAOA Hamiltonian while ensuring that the resulting circuit is SWAP-free. Our formulation also addresses the qubit allocation problem \cite{siraichi2018qubit,zhu2020exact} by allowing the optimization over permutation matrices that encode relabeling of qubits, allowing us to find the best mapping of logical qubits to physical qubits in the hardware.

We offer a theoretical guarantee on the optimum solution trade-off, and we also propose heuristics to solve the MISDP efficiently for large problem instances, making our framework practical for real-world applications.

To benchmark the robustness of our framework against a broad class of quadratic binary optimization problems, we showcase our method applied to the $k$-medoids problem, including experiments on real-world data from the index-tracking problem from finance. The $k$-medoids problem is NP-hard, and is modelled as a quadratic binary optimization problem with cardinality constraints that could usually require a fully connected network of qubits for a direct QAOA implementation. 

Our approach is, however, general and can be used in other quadratic binary optimization problems with or without cardinality constraints, such as MAXCUT \cite{sankar2024benchmarking}.

\section{Background on the problem}

For completeness, we first review the QAOA algorithm and the $k$-medoids problem applied to index-tracking. A quadratic unconstrained binary optimization (QUBO) problem can be expressed as
\begin{equation}
    \min \{ x^T C x : x \in \{0,1\}^n \}, 
\end{equation}
with $C$ an $n \times n$ matrix, which may be assumed to be symmetric. To apply QAOA, we map the objective function to a cost Hamiltonian:
\begin{equation}
    H_{C} = \sum_{i \neq j} J_{ij} Z_i Z_j + \sum_i h_i Z_i,
\end{equation}
where $Z_i$ is the Pauli-$Z$ operator on qubit $i$, and the constants $J_{ij}$ and $h_i$ are defined as
\begin{equation}
J_{ij} = \frac{C_{ij}}{4} \quad \text{and} \quad h_i = -\frac{ C_{ii}}{2} - \sum_{j \neq i} \frac{C_{ij}}{4}.
\end{equation}
By construction, for any binary vector $x \in \{0,1\}^n$,
\begin{equation}
    x^T C x = \bra{x} H_{C} \ket{x} + c_0,
\end{equation}
for a constant term $c_0$ dependent on $C$. The QAOA algorithm then proceeds by preparing a parameterized quantum state through alternating applications of the cost Hamiltonian $H_C$ and a mixing Hamiltonian, which is typically chosen as $H_M = \sum_i X_i$, where $X_i$ is the Pauli-$X$ operator on qubit $i$.

In our setting we intend to solve a quadratic binary optimization problem with cardinality constraints, which can be expressed as
\begin{equation} \label{eq:cardinality-constrained}
    \min \{ x^T C x : x \in \{0,1\}^n,\, \mathbf{1}^T x = k \}, 
\end{equation}
for some chosen $k$, and where $\mathbf{1}$ stands for the all $1$s vector. As noted in \cite{xy-mixers}, $XY$ mixers are advantageous for cardinality-constrained problems because they preserve the Hamming weight of the state, so we adopt the mixer
\begin{equation}
    H_M = \frac 1 2 \sum_{ij \in E(G)} \left( X_i X_j + Y_i Y_j \right),
\end{equation}
where $G$ is the graph that describes the hardware native capability of implementing gates involving the qubits, hereby referred to as the \emph{hardware graph}. The initial state is chosen as the Dicke State $\ket{D_k^n}$, i.e., the uniform superposition over all feasible solutions, which correspond to binary strings $x$ such that $\mathbf 1^T x = k$. It is given by
\begin{equation}
    \ket{D_k^n} = \frac{1}{\sqrt{\binom{n}{k}}} \sum_{\mathbf 1^T x = k} \ket{x}.
\end{equation}
The QAOA consists of applying $\ell$ alternating mixer and cost layers of the form
\begin{equation}
    U(\gamma, \beta) = \prod_{j=1}^\ell e^{-i \beta_j H_M} e^{-i \gamma_j H_{C}},
\end{equation}
and then measuring the resulting state to obtain a candidate solution. The parameters $\gamma$ and $\beta$ are optimized through a classical optimization loop. We do not review the details of the optimization loop here, but several approaches have been proposed in the literature, including gradient-based methods and heuristic optimization algorithms \cite{qaoa-analysis}.

As for the main point of this work, the mixer Hamiltonian $H_M$ is designed to be compatible with the hardware graph. The cost Hamiltonian $H_C$ may include interactions between qubits that are not directly connected in the hardware, leading to the need of SWAP gates to implement the required interactions, which increases circuit depth and error rates. Our framework aims to modify the cost Hamiltonian to only include interactions that can be implemented without SWAP gates, while still providing a strong approximation to the original problem.

We point out that an efficient Dicke-state preparation algorithm, proposed in \cite{dicke-state}, does not fully respect hardware topology and may introduce SWAP gates. However, unlike the cost Hamiltonian, this SWAP overhead does not scale with the number $\ell$ of QAOA layers, so we retain it in the framework.

The index-tracking problem is a portfolio optimization problem. The goal is to select a subset of assets that closely tracks the performance of a given financial index, while also satisfying constraints on the number of assets selected. The authors in \cite{marketgraph-digital-annealing} propose a solution for index tracking well-suited for quadratic binary optimization with cardinality constraints, inspired by a $k$-medoid clustering technique introduced in \cite{qubo-kmedoids}. The problem is modelled as
\begin{equation}
    \min_x \ \left\{\beta \mathbf 1^T C x - \frac{\alpha}{2} x^T C x : x \in \{0, 1\}^n, \, \mathbf{1}^T x = k\right\},
\end{equation}
where $C$ is a similarity matrix between assets, $\mathbf 1$ is the all-ones vector, and $\alpha$ and $\beta$ are empirical constants. Following \cite{qubo-kmedoids}, we define
\begin{equation} \label{eq:C}
    C_{ij} = 1 - \exp \left(-(1 - \text{Corr}(i,j)) \right).
\end{equation}
This problem searches for a subset $x$ of $k$ assets that are mutually dissimilar (large $x^T C x$) while remaining representative of the full universe (small $\mathbf 1^T C x$). Since $x \in \{0,1\}^n$, the linear term can be absorbed into the diagonal of the quadratic form. As $C$ is symmetric, define
\begin{equation}\label{eq:hatC}
    \hat C = \beta \cdot \text{Diag}(C \cdot \mathbf 1) - \frac{\alpha}{2} C,
\end{equation}
where $\text{Diag}(C \cdot \mathbf 1)$ is the diagonal matrix whose $i$-th diagonal entry is the $i$-th entry of the vector $C \cdot \mathbf 1$. Then the problem is equivalently written as
\begin{equation}
    \min_x \ \left\{ x^T \hat C x : x \in \{0, 1\}^n, \mathbf{1}^T  x= k \right\}.
\end{equation}
Unlike \cite{marketgraph-digital-annealing}, we do not model the problem as a QUBO with penalty terms forcing $\mathbf 1^T x = k$. Rather, we enforce the cardinality constraint through the quantum setting: the $k$-excitation subspace contains the initial state and is invariant under both the cost and the mixer Hamiltonian.

\section{MISDP formulation for SWAP-free QAOA}

Our goal in this section is to formalize the problem of designing a SWAP-free QAOA circuit as a Mixed-Integer Semidefinite Program (MISDP). The key idea is to alter the matrix $C$ defined in \eqref{eq:cardinality-constrained} so that the cost Hamiltonian $H_C$ only includes interactions between qubits that are directly connected in the hardware graph, where onwards defined as $G=(V_G,E_G)$. This means that we want to find a modified cost matrix $\tilde{C}$ such that $\tilde{C}_{ij} = 0$ for all $ij \notin E_G$, while also ensuring that $\tilde{C}$ is as close as possible to the original cost matrix $C$ in terms of the objective function value in \eqref{eq:cardinality-constrained}. 

Let $A_G$ denote the $01$-symmetric adjacency matrix of the graph $G$, which we assume has no vertex loops. Let $L_G$ denote the Laplacian matrix of the graph, i.e., $L_G = D_G - A_G$, where $D_G$ is diagonal with the degrees of the vertices in their respective diagonals. Let $\overline{G}$ denote the complement graph, thus $A_G = \mathbf{1}\mathbf{1}^T - A_{\overline{G}} - I$, where $I$ is the identity matrix. A permutation matrix is a $01$-square matrix with exactly one entry of $1$ in each row and each column, and $0$s elsewhere. Note that permutation matrices are precisely the integral doubly stochastic matrices. For matrices $M$ and $N$ of the same order, recall that $M \circ N$ denotes their Hadamard (entrywise) product, i.e., $(M\circ N)_{ij} = M_{ij} N_{ij}$.

We can then express the problem of finding $\tilde{C}$ as the following mixed-integer semidefinite optimization problem. This formulation features a solution that enforces the hardware topology while tackling the qubit allocation problem, as the permutation matrix $P$ encodes a relabeling of the qubits.

\begin{align}
    \min_{\lambda,X,P} \ \  & \lambda \label{eq:misdp}\\
    \text{s.t. } & \lambda \in \mathbb{R}_+, \, X \in \mathbb{R}^{V_G \times V_G}, \, P \in \mathbb{R}^{V_G \times V_G}, \nonumber \\ 
    & \lambda I \succeq X -  \hat C  \succeq -\lambda I, \nonumber \\
    & X \circ (P \cdot A_{\overline{G}} \cdot P^T) = 0, \nonumber \\
    & P \text{ is a permutation matrix}. \nonumber
\end{align}

We discuss the formulation \eqref{eq:misdp} in detail below.

\begin{proposition}\label{prop:misdp}
    Formulation \eqref{eq:misdp} is a polynomially-sized mixed integer semidefinite program.
\end{proposition}

\begin{proof}
    We begin by examining an equivalent form of \eqref{eq:misdp}. Note that the constraint $X \circ (P \cdot A_{\overline{G}} \cdot P^T) = 0$ is equivalent to $(P^T \cdot X \cdot P) \circ A_{\overline{G}} = 0$, and by relabeling $(P^T \cdot X \cdot P)$ as $Y$, we note that
\[ \lambda I \succeq X -  \hat C  \succeq -\lambda I \iff \lambda I \succeq Y -  (P^T \cdot \hat C \cdot P) \succeq -\lambda I.
\]
It follows that
\[ (P^T \cdot \hat C \cdot P)_{ij} = \sum_{k,l} P_{ki} P_{\ell j} \hat C_{k \ell},
\]
and with the intent of replacing $P_{ki} P_{\ell j}$ by variables $q_{ki  \ell j}$, we may conclude that $P$ is a permutation matrix if and only if
\begin{enumerate}
    \item $\sum_k q_{kiki} = 1$ for all $i$,
    \item $\sum_i q_{kiki} = 1$ for all $k$,
    \item $q_{ki \ell j} \leq q_{kiki}, q_{ki\ell j} \leq q_{\ell j \ell j}$ for all $i,j,k,\ell$,
    \item $q_{ki \ell j} \geq q_{kiki} + q_{\ell j \ell j} - 1$ for all $i,j,k,\ell$,
    \item $0 \leq q_{ki \ell j} \leq 1$ for all $i,j,k,\ell$,
    \item $q_{kiki} \in \{0,1\}$ for all $i,k$.
\end{enumerate}
Thus, we may rewrite \eqref{eq:misdp} as the following optimization problem. It is a mixed-integer semidefinite program, as it features a linear objective function, linear matrix inequalities, and linear constraints on the variables, including integrality constraints on some variables, and it has a number of variables and constraints that is polynomial in the size of the input data.
\begin{align}
    \min_{\lambda,Y,q} \ \  & \lambda \label{eq:misdp2}\\
    \text{s.t. } & \lambda \in \mathbb{R}_+, \, Y \in \mathbb{R}^{V_G \times V_G}, \, q \in \mathbb{R}^{V_G \times V_G \times V_G \times V_G}, \nonumber \\ 
    & \lambda I \succeq Y -  q(\hat C)  \succeq -\lambda I, \nonumber \\
    & Y \circ A_{\overline{G}} = 0, \nonumber \\
    & q \text{ satisfies conditions (i) to (vi) above}, \nonumber \\
    & q(\hat C)_{ij} = \sum_{k,l} q_{ki \ell j} \hat C_{k \ell} \text{ for all } i,j. \nonumber
\end{align}
\end{proof}

The optimal value found for $\lambda$ also bounds the deviation from the optimal value of the quadratic optimization problem \eqref{eq:cardinality-constrained} when using $\tilde C$ instead of $\hat C$ to find an optimal solution. In particular, if $x$ is an optimal solution for the original problem and $\tilde x$ is an optimal solution for the modified problem, then
\begin{equation}
     x^T \hat C x \leq \tilde x ^T  \hat C \tilde x \leq \tilde x^T \tilde C \tilde x + \lambda k \leq x^T \tilde C x + \lambda k \leq x^T \hat C x + 2 \lambda k.
\end{equation}

A solution for the MISDP \eqref{eq:misdp} has $\lambda = 0$ if and only if $\tilde C =  \hat C$. As a consequence, one can reduce the subgraph isomorphism problem (see \cite{garey1977rectilinear}) to the decision version of \eqref{eq:misdp}, yielding the following result.

\begin{proposition}
    The decision problem which asks whether there exists a solution of \eqref{eq:misdp} whose optimal value is at most a given $\lambda_0$ is NP-complete. \qed
\end{proposition}

An important feature of the formulation \eqref{eq:misdp2} is that it does not require us to explicitly enumerate all possible permutation matrices $P$, rather, it allows us to use mixed-integer semidefinite programming techniques, in particular, branch-and-bound algorithms, to search for the optimal permutation matrix implicitly through the variables $q$.

Upon fixating the matrix $P$ (and henceforth omitting it), the problem \eqref{eq:misdp} is an SDP, which admits a dual formulation:
\begin{align}
    \max_{M,N} \ \  & \langle M-N, \hat C \rangle \label{eq:misdp3}\\
    \text{s.t. } & M,N \in \mathbb{R}^{V_G \times V_G}, \nonumber \\ 
    & \langle M+N , I \rangle \leq 1, \nonumber \\
    & (M-N) \circ (I + A_G) = 0, \nonumber \\
    & M,N \succeq 0. \nonumber
\end{align}

The primal formulation \eqref{eq:misdp} (with $P = I$) satisfies Slater's condition (see \cite{shapiro2000duality}) with the feasible solution $(\lambda,X) = (\lambda_0, 0)$ for any $\lambda_0 > \lVert C \rVert_{\text{op}}$, hence strong duality holds between \eqref{eq:misdp} and \eqref{eq:misdp3}, and the optimal values of both problems are equal. We may also note that if $Y = M-N$ in \eqref{eq:misdp3}, then among all decompositions of $Y$ as a difference of positive semidefinite matrices $M'-N'$, the one that minimizes $\langle M'+N', I \rangle$ is precisely the decomposition given by the positive and negative parts of $Y$, i.e., $M = Y_+$ and $N = Y_-$. Let $|Y| = Y_+ + Y_-$. Thus, we may rewrite \eqref{eq:misdp3} as
\begin{align}
    \max_{Y} \ \  & \langle Y, \hat C \rangle \label{eq:misdp4}\\
    \text{s.t. } & Y \in \mathbb{R}^{V_G \times V_G}, \nonumber \\ 
    & \langle |Y| , I \rangle \leq 1, \nonumber \\
    & Y \circ (I + A_G) = 0. \nonumber
\end{align}

Given the target hardware graph $G$ (assuming $P$ has been implicitly determined) and the cost matrix $\hat C$, we may upper bound the optimum value of \eqref{eq:misdp} by the feasible solution
\[ X = \hat C \circ (I+A_G), \quad \lambda = \lVert \hat C \circ A_{\overline{G}} \rVert_{\text{op}} = \frac{\alpha}{2} \lVert C \circ A_{\overline{G}} \rVert_{\text{op}}.
\]
The following result offers an alternative bound, which is independent of the permutation matrix, and is instead expressed in terms of the Lovász number of the graph $G$ (see \cite{lovasz1979shannon}).  

\begin{theorem} \label{thm:upper-bound}
    Given $\hat C$ as in \eqref{eq:hatC}, define $L = \text{Diag}(C \cdot \mathbf{1}) - C$. The optimum value $\lambda^*$ of \eqref{eq:misdp} satisfies
    \[\lambda^* \leq \left(1 - \frac{1}{\vartheta(G)} \right) \frac{\alpha}{2} \langle L,L \rangle.,\]
    where $\vartheta(G)$ is the Lovász number of the graph $G$.
\end{theorem}
\begin{proof}
    A special feature of our original problem \eqref{eq:hatC} is that $C \geq 0$ (entrywise), hence if $L = \text{Diag}(C \cdot \mathbf{1}) - C$, then 
    \begin{equation} \label{eq:hatC-decomposition}
        \hat C = \left( \beta - \frac{\alpha}{2}\right) D + \frac{\alpha}{2} L,
    \end{equation}
    where $D$ is diagonal and $L$ is a Laplacian matrix, i.e., a matrix with nonpositive off-diagonal entries and zero row sums. In particular, $L$ is positive semidefinite. In order to give the desired upper bound, let $Y$ be an optimal solution for \eqref{eq:misdp4}. Recall that, for any two matrices $S$ and $T$,
    \begin{equation} \label{eq:inner-product}
       \langle S,T \rangle = \text{sum of all entries}(S \circ T).
    \end{equation}
    Let $A = A_{\overline{G}}$. Then
\begin{align}
    \langle Y, (2/\alpha)\hat C \rangle^2 & = \langle Y , L \rangle^2 \label{eq:proof1} \\ 
    & = \langle Y \circ A, L \circ A \rangle^2 \label{eq:proof2} \\ 
    & \leq \langle Y \circ A, Y \circ A \rangle \langle L \circ A, L \circ A \rangle  \label{eq:proof3} \\ 
    & = \langle A, Y \circ Y \rangle \langle A, L \circ L \rangle \label{eq:proof4} \\ 
    & \leq \left(1 - \frac{1}{\chi_\text{vec}(\overline{G})} \right)^2 \langle J, Y \circ Y \rangle \langle J, L \circ L \rangle \label{eq:proof5} \\ 
    & = \left(1 - \frac{1}{\chi_\text{vec}(\overline{G})} \right)^2 \langle Y, Y \rangle \langle L,L \rangle \label{eq:proof6} \\ 
    & \leq \left(1 - \frac{1}{\chi_\text{vec}(\overline{G})} \right)^2 \langle L,L \rangle \label{eq:proof7}
\end{align}
Line \eqref{eq:proof1} follows from the fact that $Y$ has zero diagonal and equation~\eqref{eq:hatC-decomposition}. Line \eqref{eq:proof2} follows from the fact that $Y \circ (I + A_G) = 0$ and equation~\eqref{eq:inner-product}. Line \eqref{eq:proof3} is the Cauchy-Schwarz inequality. Line \eqref{eq:proof4} follows from equation~\eqref{eq:inner-product}. Line \eqref{eq:proof5} follows from \cite[Lemma 5]{coutinho2024conic}. Line \eqref{eq:proof6} follows from equation~\eqref{eq:inner-product}. Finally, line \eqref{eq:proof7} follows from the constraint $\langle |Y|, I \rangle \leq 1$. For our last step, we recall that $\chi_\text{vec}(\overline{G}) \leq \vartheta(G)$ (see for instance \cite{karger1998approximate}).
\end{proof}

\section{Heuristic solutions to the MISDP}
\label{sec:heuristic-results}

We compared seven strategies for computing the permutation matrix \(P\) before solving the MISDP: \emph{Perron Disconnected}, \emph{Perron Connected}, \emph{Laplacian Connected}, \emph{Completely/Partially Random Disconnected}, and \emph{Completely/Partially Random Connected}. We discuss their formulation below, followed by the experimental results.

\subsection{Heuristic strategies}

To handle large instances of the MISDP \eqref{eq:misdp}, we consider heuristic strategies in which the permutation matrix \(P\) is fixed in advance, so that only a single SDP needs to be solved for each candidate permutation.

Let \(G\) be the hardware graph on \(n\) vertices, and let \(A_G\) be its adjacency matrix. Let \(v\) be a Perron eigenvector of \(A_G\), that is, an eigenvector associated with its largest eigenvalue. Likewise, let \(u\) be a Perron eigenvector of \(C\) (which plays the role of the adjacency matrix in equation \eqref{eq:hatC}). Define \(\pi\) as a permutation such that
\[
v_{\pi(1)} \ge v_{\pi(2)} \ge \cdots \ge v_{\pi(n)},
\]
and define \(\sigma\) analogously from the coordinates of \(u\). Thus, \(\pi\) and \(\sigma\) rank the vertices of \(G\) and the indices of \(\hat C\), respectively, according to Perron centrality (see for instance \cite{martin2014localization}).

Our first heuristic matches the most central indices of \(C\) to the most central vertices of \(G\), without imposing any connectivity requirement. More precisely, we define \(P\) as the permutation matrix corresponding to \(\sigma\pi^{-1}\), and solve the MISDP with this choice fixed.

This approach is motivated by the idea that large Perron-vector coordinates identify the most influential vertices in each structure. However, a possible drawback is that the vertices of \(G\) selected in this way need not induce a connected subgraph. Consequently, indices \(i,j\) with large values of \(\hat C_{ij}\) may be mapped to non-adjacent qubits, forcing the corresponding entries of the embedded matrix to vanish.

To mitigate this issue, we present a connectivity-aware variant. The heuristic starts by assigning the most important index of \(C\), namely \(\sigma(1)\), to the most important vertex of \(G\), namely \(\pi(1)\). Then, at each subsequent step, it assigns the next index in the order induced by \(\sigma\) to the highest-ranked unused vertex lying in the neighborhood of the already selected set.

Formally, let \(S\subseteq V(G)\) denote the set of already assigned vertices. Initialize
\[
P[\pi(1)] \leftarrow \sigma(1), 
\qquad
S \leftarrow \{\pi(1)\}.
\]
For \(i=2,\dots,n\), choose
\[
v \in \arg\min \{ \pi^{-1}(w) : w \in N(S)\setminus S\},
\]
and set
\[
P[v] \leftarrow \sigma(i), 
\qquad
S \leftarrow S \cup \{v\}.
\]
Since each newly selected vertex lies in \(N(S)\), the set \(S\) remains connected throughout the procedure. In this way, the most important indices of \(C\) are mapped to a connected region of the hardware graph.

\begin{algorithm}
    \SetAlgoLined
    \KwData{Permutations \(\pi,\sigma\)}
    \KwResult{Permutation \(P\)}
    \(P[\pi(1)] \leftarrow \sigma(1)\)\;
    \(S \leftarrow \{\pi(1)\}\)\;
    \For{\(i = 2,\dots,n\)}{
        \(v \leftarrow \arg\min \{\pi^{-1}(w) : w \in N(S)\setminus S\}\)\;
        \(P[v] \leftarrow \sigma(i)\)\;
        \(S \leftarrow S \cup \{v\}\)\;
    }
    \caption{Connectivity-aware heuristic permutation}
    \label{alg:heuristic}
\end{algorithm}

The \emph{Laplacian Connected} heuristic compares the order induced by the eigenvector corresponding to the largest eigenvalue of \(L_G\), the Laplacian matrix of the graph, and that of \(\hat C\).

To benchmark the eigenvector-based strategies, we also consider four random baselines. In both cases, we generate \(m\) random permutations and solve the corresponding SDP for each of them, keeping the solution with minimum value of \(\lambda\).

In \emph{Completely/Partially Random Disconnected}, the permutation is sampled uniformly at random, with no connectivity requirement. In \emph{Completely/Partially Random Connected}, the permutation is constructed incrementally in the same spirit as Algorithm~\ref{alg:heuristic}, but replacing the Perron-based priority rule by random choices while preserving connectivity of the selected vertices.

The difference between completely and partially random is that the former generates a random permutation matrix for both the graph matrix and the $\hat C$ matrix to ensure a completely random matching between entries, whilst the latter uses the Perron eigenvector to sort the matrix $\hat C$. No significant differences were found in the results for the four procedures.

\subsection{Heuristic comparison methodology}

We first worked to compare the heuristics for small instances of the problem, where the Brute Force procedure is computationally available. This allows us to evaluate how close the heuristics come to minimizing the MISDP objective, and how this translates into optimality gaps for the original problem.

To benchmark the heuristics against the Brute Force solution, we restrict to instances with \(n=8\), for which all \(n!\) permutations can be enumerated. We generated \(22{,}500\) random graph instances, with densities \(d=0.1,0.2,\dots,0.9\) distributed uniformly across the sample. For each instance, a correlation matrix of size \(6\times 6\) was randomly sampled from the S\&P 500 database, and the portfolio size was fixed at \(k=2\).

To evaluate the quality of the proposed heuristics, we use the \emph{optimality gap}, a standard metric for minimization problems, defined by
\[
    \text{optimality gap}
    =
    \frac{
        \text{algorithm value}
        -
        \text{optimal value}
    }{
        \text{optimal value}
    }.
\]

The algorithm value is defined as follows: for each matrix $X$ resulting from the SDP in equation \eqref{eq:misdp} coming from an heuristic, we brute force sort all the possible combinations of $k$ choices of assets. We then select the top $1\%$ $k$-choices on the matrix $X$ and extract the optimal value among them applied to $\hat C$. We chose to do so as one cannot extrapolate the fitness of a solution in the altered matrix to the original one, even though we have a theoretical approximation bound, so we consider a pool of the best candidates.

We also report the normalized $\lambda$ objective value,
\[
    \lambda / \|\hat C\|,
\]
where \(\lambda\) is the objective value returned by the MISDP relaxation.

As mentioned, on small instances we also compute a \emph{Brute Force} benchmark, which evaluates all permutations \(P\) and returns the one with minimum \(\lambda\).

\subsection{Heuristic performance for small graphs}

Figure~\ref{fig:n_8_heuristics} reports the normalized lambda value and the optimality gap for five heuristics---Perron Disconnected, Perron Connected, Laplacian Connected, Partially Random Disconnected, and Partially Random Connected---together with the Brute Force benchmark.

\begin{figure}[ht]
    \centering
    \begin{subfigure}[b]{0.49\textwidth}
        \centering
        \includegraphics[width=\textwidth]{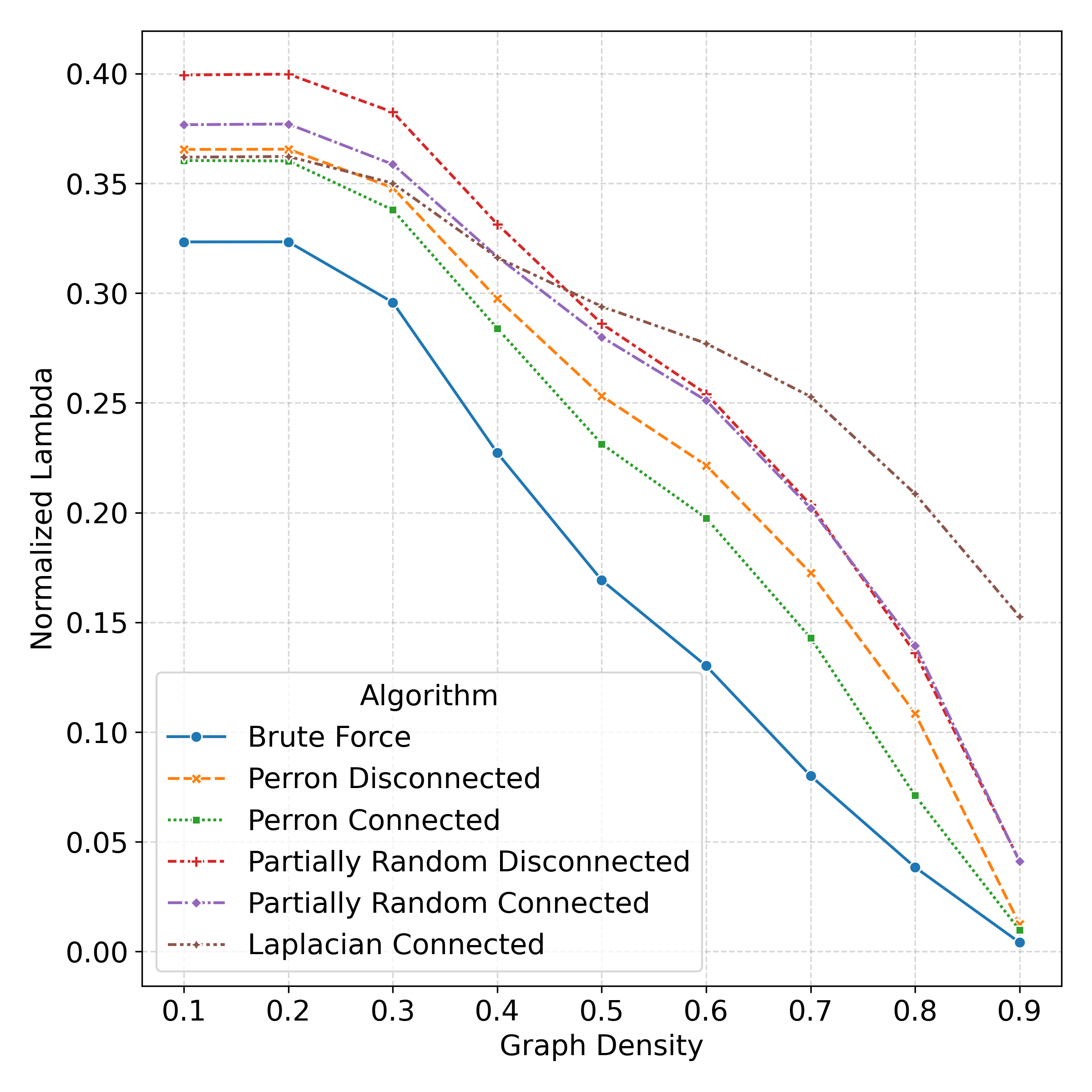}
    \end{subfigure}
    \hfill
    \begin{subfigure}[b]{0.49\textwidth}
        \centering
        \includegraphics[width=\textwidth]{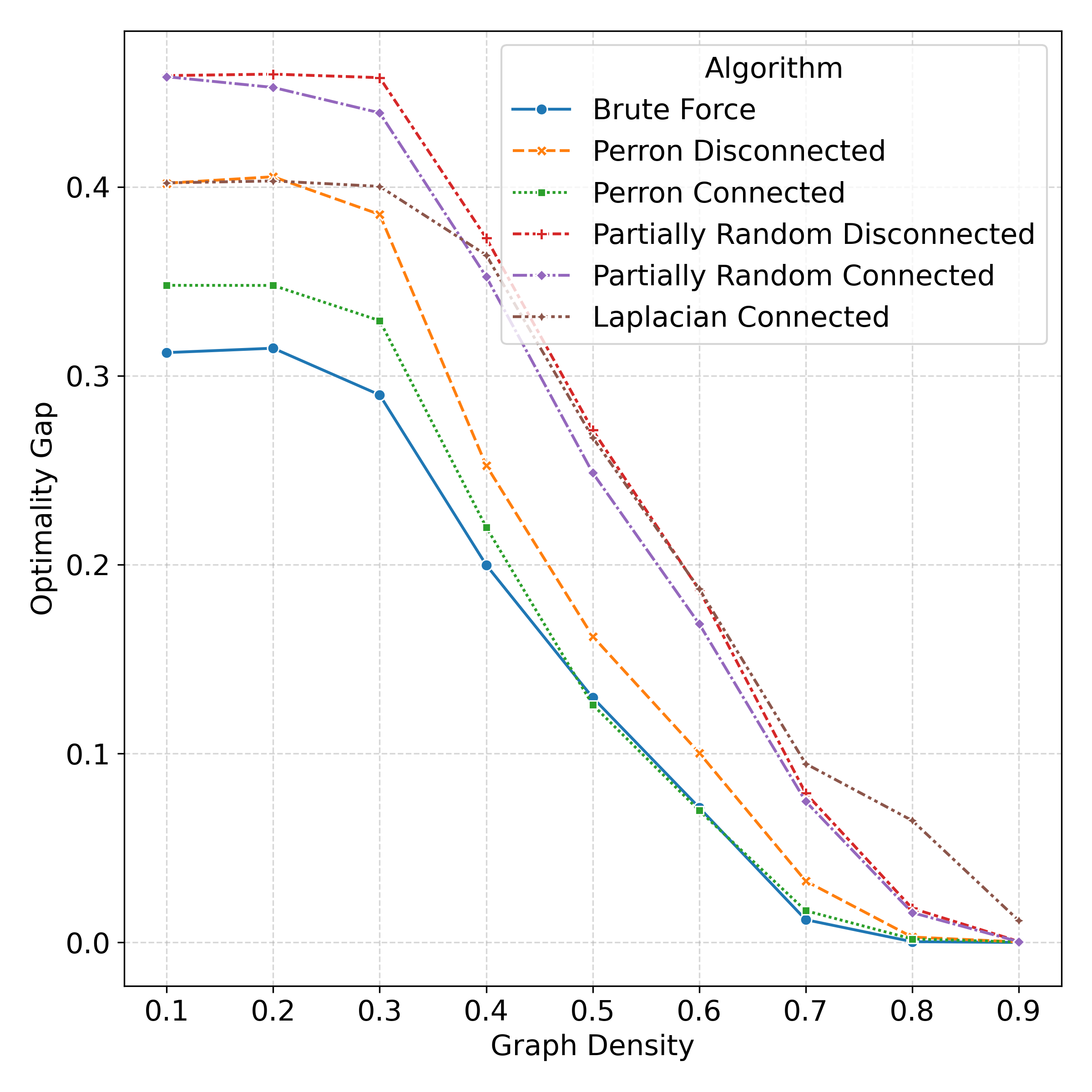}
    \end{subfigure}
    
    \caption{Comparison of heuristics. The plots show how normalized lambda and optimality gap vary with graph density for \(n=8\).}
    \Description{Comparison of heuristics for small graphs. The left plot shows the normalized lambda values, and the right plot shows the optimality gap for different heuristics as a function of graph density for \(n=8\).}
    \label{fig:n_8_heuristics}
\end{figure}

The results indicate that the adjacency Perron-based heuristics generally outperform the random ones on these small instances. Interestingly, the Laplacian Connected heuristic is one of the best performing for low densities values, but its results do not improve as considerably as for the other heuristics, becoming the worst performing overall for high density values. Among the random strategies, enforcing connectivity improves performance, since Random Connected consistently dominates Random Disconnected. In contrast, the difference in the optimality gap between Perron Connected and Perron Disconnected is even more pronounced for low density intervals.

As expected, the Brute Force procedure, which is computationally available at such small instances, yields the lowest normalized lambda values overall, which suggests that the proposed heuristics still leave room for improvement when the sole objective is to minimize the MISDP relaxation. However, the plots also show that lower values of \(\lambda\) do not always translate into significant difference of optimality gaps for mid to high density values.

This indicates that solving the MISDP more accurately is not, by itself, sufficient to guarantee better portfolio solutions. In particular, it suggests that polynomial-time heuristics may still achieve near-optimal solutions for the original problem, even when they do not come close to minimizing the MISDP objective.

\section{Comparison against QAOA with SWAPs}

In this section we evaluate our approach against a baseline representing ideal QAOA execution under SWAP-induced noise regime.

Our findings are summarized in Figure~\ref{fig:opt_gap}, which compares the optimality gap of our heuristics with that of an ideal QAOA with SWAP errors, as a function of the problem size \(n\). 

\begin{figure}[htb]
    \centering
    \centering
    \begin{subfigure}[b]{0.49\textwidth}
        \centering
        \includegraphics[width=\textwidth]{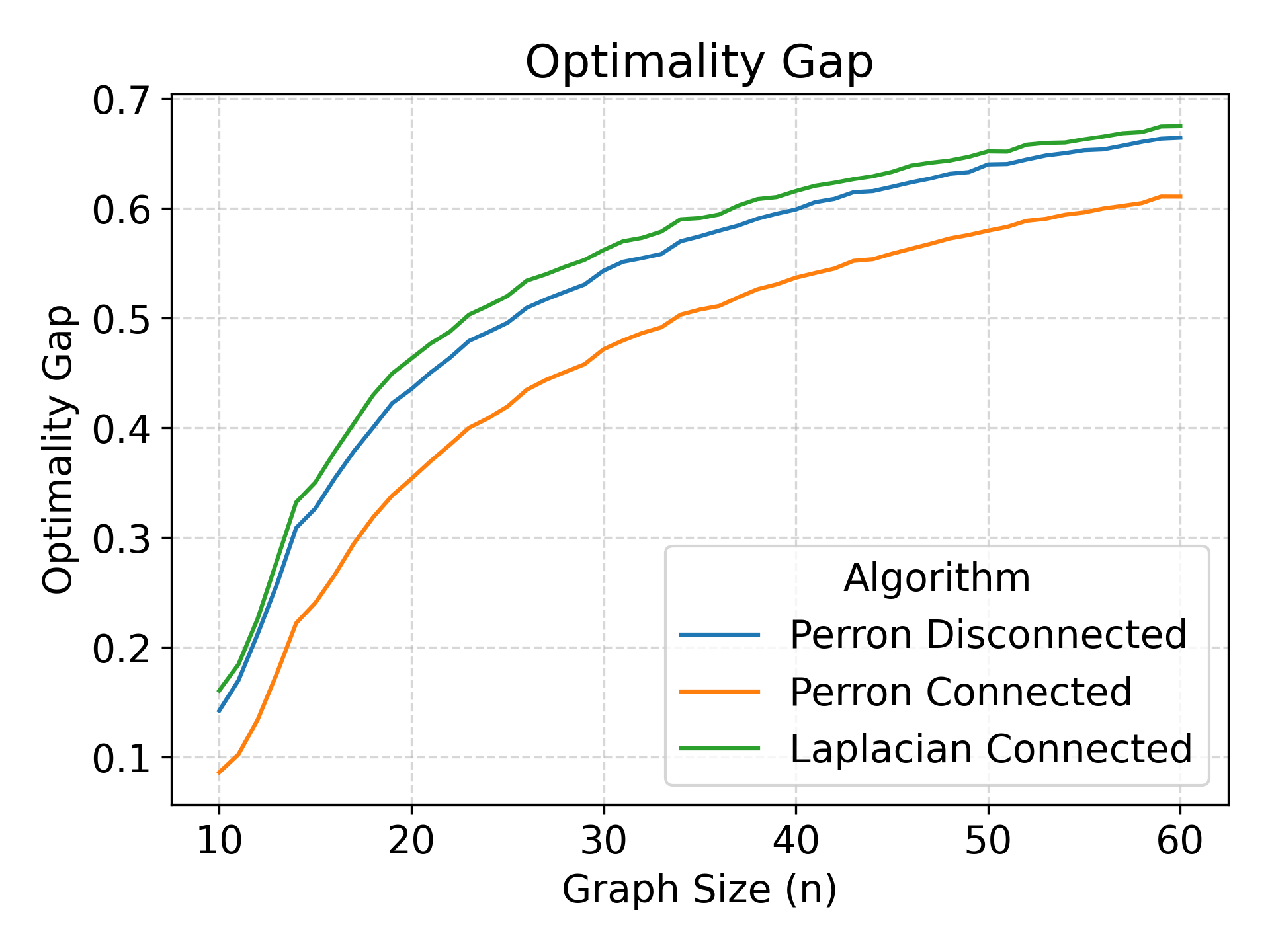}
    \end{subfigure}
    \hfill
    \begin{subfigure}[b]{0.49\textwidth}
        \centering
        \includegraphics[width=\textwidth]{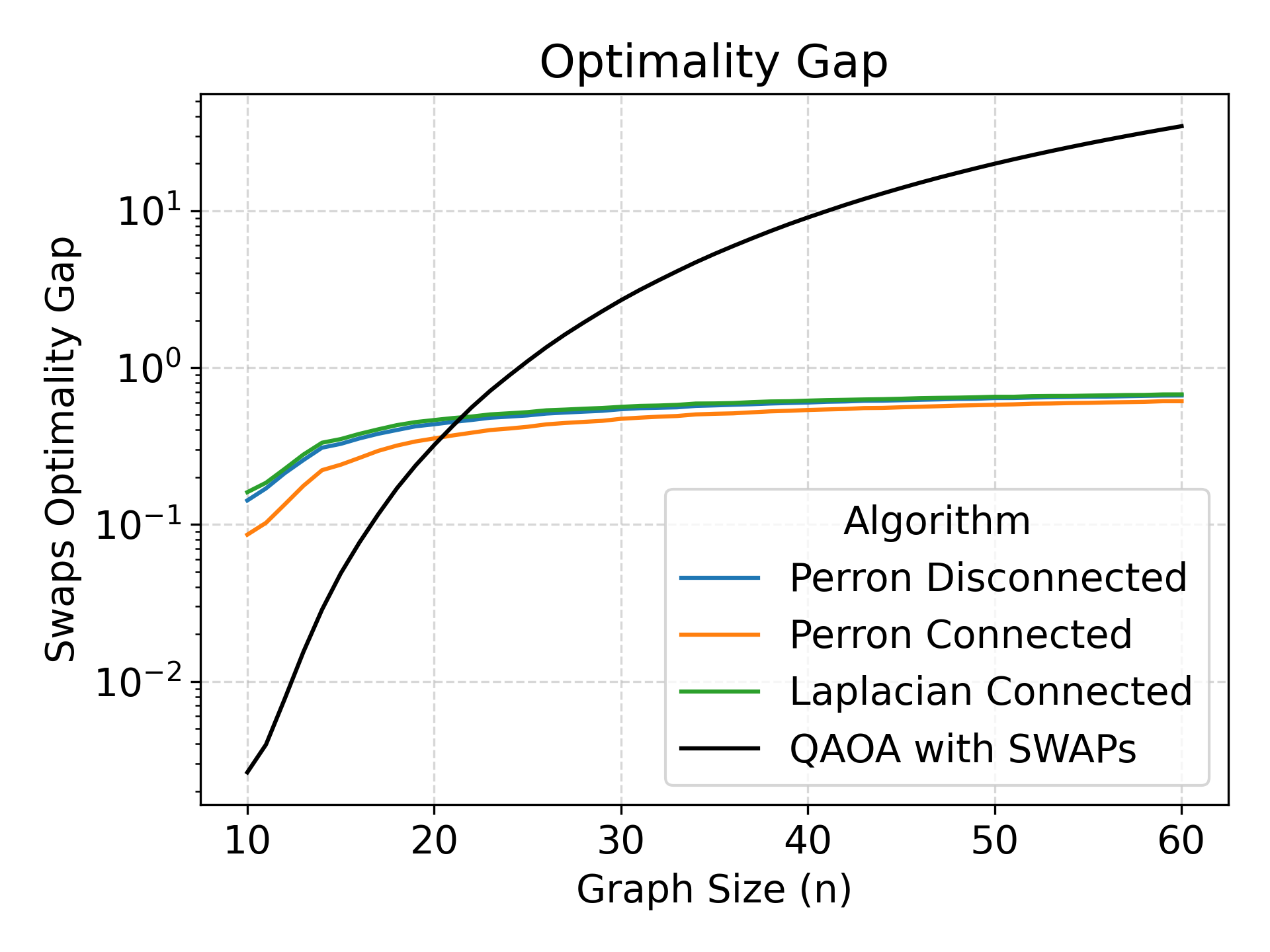}
    \end{subfigure}
    \caption{Optimality gap as a function of \(n\) for four algorithms: Perron Connected, Perron Disconnected, Laplacian Connected, and ideal QAOA with SWAP errors.}
    \Description{Comparison of optimality gaps for different algorithms as a function of problem size \(n\). The plot shows the performance of Perron Connected, Perron Disconnected, Laplacian Connected, and ideal QAOA with SWAP errors.}
    \label{fig:opt_gap}
\end{figure}

For \(n>20\), our methods significantly outperform the SWAP-based baseline. This suggests that, as the problem size grows, avoiding SWAPs, even at the cost of imprecise objective functions, becomes more beneficial than relying on standard qubit-allocation strategies followed by noisy execution. For smaller values of \(n\), the gap is less pronounced. A plausible explanation is that the number of required SWAPs, which grows roughly quadratically with \(n\), is still relatively small in this regime, so the performance degradation caused by SWAP noise remains limited.

\subsection{Experimental details}

For the comparison against QAOA with SWAP-induced noise, we generated random connected graphs with number of vertices ranging from \(n=10\) to \(n=60\). For each value of \(n\), we generated \(4000\) graph instances with edge density \(d=0.5\). For each graph, we sampled a correlation matrix from the S\&P 500 database with number of assets \(m=n-2\), and fixed the portfolio size at \(k=4\).

A difficulty in evaluating qubit-allocation procedures for QAOA is that poor performance may come either from quantum errors that increase with a poor allocation, or from imperfect optimization of the classical parameters \(\beta\) and \(\gamma\) within the QAOA setting. To isolate the effect of qubit allocation, we introduce a metric based on an idealized QAOA convergence assumption.

Let \(\tilde C\) be the approximation of \(\hat C\) produced by one of our heuristics. We assume that QAOA returns a vector \(z\) solving
\[
    \min \; z^T \tilde C z
    \qquad
    \text{s.t. }
    z \in \{0,1\}^n,\quad
    z^T \mathds{1}=k.
\]
The corresponding \emph{algorithm value} is then measured in the original objective, namely \(z^T \hat C z\).

To define a baseline corresponding to standard QAOA with SWAPs, we assume perfect convergence to the ground state of the cost Hamiltonian, but subject to SWAP-induced hardware noise. Specifically, we model the circuit output by the depolarizing channel
\[
    (1-p)\ket{x^*}\bra{x^*} + p \frac{I}{2^n},
\]
where \(x^*\) is the optimal binary solution of the quadratic binary optimization instance and \(p\) is an error probability determined by the number of inserted SWAP gates. In our experiments, we set the CNOT error rate\footnote{For instance, IBM's Miami 2 reported an error rate of 0.638\% (retrieved in February 03, 2026) \cite{IBM-cnot-error}} to \(0.33\%\). Since each SWAP gate is implemented using three CNOTs, we define
\[
    p = 1 - (1-\text{CNOT error rate})^{3\cdot(\text{SWAP count})}.
\]

The SWAP count is estimated using Qiskit's transpilation API, taking the expected number of SWAPs required for one QAOA layer for the given cost Hamiltonian and a random parameter \(\gamma_i\in[0,2\pi)\). Under this model, the expected solution value becomes
\begin{align*}
    &(1-p)(x^*)^T \hat C x^* + \frac{p}{2^n}\sum_x x^T \hat C x \\
    &=
    (1-p)\,\text{optimum value}
    + \frac{p}{4}\bigl(\operatorname{tr}(\hat C)+\operatorname{sum}(\hat C)\bigr).
\end{align*}

\section{Conclusion}

This work proposes a hardware-aware approach to QAOA for NISQ devices. Instead of treating limited qubit connectivity as a compilation issue to be resolved by transpilation, we incorporate the hardware graph directly into the design of the cost Hamiltonian. The resulting framework replaces the original quadratic objective by a hardware-compatible approximation, chosen to eliminate SWAP gates in the cost-layer implementation while remaining close to the target problem.

This approximation task is formulated as a mixed-integer semidefinite program (MISDP), which optimizes both the choice of interactions in the cost Hamiltonian and the mapping of logical variables to physical qubits, aspects which are often studied separately. To solve the MISDP, we propose heuristics based on the spectral properties of the problem and of the hardware graph. 

Our experiments indicate that this trade-off can be advantageous in realistic regimes. In particular, the MISDP-based framework, combined with the proposed permutation heuristics, seemed to yield better results compared with a baseline representing ideal QAOA optimization under SWAP-induced noise. This suggests that, on sparse hardware topologies and under NISQ conditions, preserving the exact objective function is not always preferable if doing so requires a large number of additional two-qubit gates.

Although we focused on the index-tracking problem through a \(k\)-medoids-inspired formulation, the proposed framework is general and can be extended to other quadratic binary optimization problems, including settings with or without cardinality constraints. 

There are several natural directions for future work. On the theoretical side, it would be valuable to derive sharper guarantees connecting the MISDP approximation parameter to the loss in the original objective. On the algorithmic side, more refined heuristics could be investigated, for example by combining spectral information with local search or other combinatorial improvement procedures. On the experimental side, an important next step is to test the framework on real quantum hardware, where device-specific noise and compilation effects may further affect the balance between exactness of the Hamiltonian and SWAP avoidance.

\section*{Acknowledgements}
We acknowledge grant support from CAPES, FAPEMIG and CNPq. This work was funded by a cooperation agreement between the Federal University of Minas Gerais and Banco Inter. We acknowledge Cristiano Arbex, Rodrigo Chaves, Luan Costa, Pedro Faria, Bruno Grossi, Mathias Oliveira, Henrique Soares and Reinaldo Vianna for helpful discussions and suggestions.

\bibliographystyle{plain}
\bibliography{references}

\end{document}